\documentclass[journal,10pt]{IEEEtran}
\usepackage{amsmath}
\usepackage{graphicx}
\usepackage{indentfirst}
\usepackage{cases}
\usepackage[noadjust]{cite}
\usepackage{filecontents}

\usepackage{tabularx,booktabs, caption}
\usepackage{makecell}

\newcommand{\blue}{\textcolor{black}}

\newcolumntype{C}{>{\centering\arraybackslash}X} 
\usepackage{fancyhdr}
\usepackage[justification=centering]{caption}
\usepackage{amsmath,amssymb,mathtools,bm,etoolbox}
\usepackage{color}
\usepackage{array}
\usepackage{mathtools}
\usepackage[british]{babel}
\usepackage{csquotes}
\usepackage{nccmath}
\usepackage{gensymb}
\usepackage[shortlabels]{enumitem}
\usepackage[nodisplayskipstretch]{setspace}
\usepackage[ruled,linesnumbered]{algorithm2e}
\usepackage[font=footnotesize]{caption}
\usepackage{booktabs}

\SetLabelAlign{bibright}{\hss\llap{[#1]}}
\newcounter{mynum}

\newtheorem{proposition}{Proposition}

\usepackage{hyperref}
\hypersetup{
     colorlinks   = true,
     citecolor    = red,
     linkcolor    = red,
     urlcolor     = black
}

\allowdisplaybreaks 

\title{\huge OTFS-Enabled Delayed SINR-Feedback Power Control for Reliable and Fair High-Mobility UAV Communications}
\begin{document}
    \author{Thuan Van Le, Nguyen Cong Luong, Trong-Dai Hoang,
    Vo Nguyen Quoc Bao,~\IEEEmembership{Senior Member,~IEEE},
    Thien Huynh-The,~\IEEEmembership{Senior Member,~IEEE},
    Xingwang Li,~\IEEEmembership{Senior Member,~IEEE},
    and Ngo Hoang Tu,~\IEEEmembership{Member,~IEEE}
    \thanks{T. V. Le is with the Faculty of Electrical and Electronic Engineering, Phenikaa School of Engineering, Phenikaa University, Hanoi 12116, Vietnam (e-mail: thuan.levan@phenikaa-uni.edu.vn).}
    \thanks{N. C. Luong is with the Phenikaa School of Computing, Phenikaa University, Hanoi 12116, Vietnam (e-mail: luong.nguyencong@phenikaa-uni.edu.vn).}
    \thanks{T.-D. Hoang is with the Global Big Data Technologies Centre, University of Technology Sydney, Ultimo, NSW 2007, Australia (e-mail: dai.t.hoang@student.uts.edu.au).}
    \thanks{V. N. Q. Bao and N. H. Tu are with the Faculty of Information Technology, Van Lang School of Technology, Van Lang University, Ho Chi Minh City 70000, Vietnam (e-mail: bao.vnq@vlu.edu.vn, tu.nh@vlu.edu.vn). (\textit{Corresponding author: Ngo Hoang Tu.})}
    \thanks{T. Huynh-The is with the Department of Electronics and Information Engineering, Ho Chi Minh City University of Technology and Engineering, Ho Chi Minh City 71307, Vietnam (e-mail: thienht@hcmute.edu.vn).}
    \thanks{X. Li is with the School of Physics and Electronic Information Engineering, Henan Polytechnic University, Jiaozuo 454003, China (e-mail: lixingwangbupt@gmail.com).}
    }%
\maketitle

\begin{abstract}
This paper develops a power control framework driven by delayed signal-to-interference-plus-noise ratio (SINR) feedback for orthogonal time frequency space (OTFS) unmanned aerial vehicle (UAV) communications operating under high mobility, with reliability and fairness as the primary design targets.
\blue{A base station with a uniform linear array serves several UAVs on a common OTFS frame, while the path delays, Doppler shifts, and inter-UAV interference are determined by the three-dimensional propagation geometry and the base-station array response rather than by a postulated coupling model.}
In place of instantaneous channel state information, the proposed controller refreshes the transmit-power vector from delayed SINR measurements alone, which matches the practical limitations of fast-fading aerial links.
\blue{A prediction--smoothing--projection rule mixes a reliability share, a fairness share and a spectral-efficiency share, each normalized separately, so that the utility weights control the closed loop directly.}
\blue{Simulations show that the effective SINR of OTFS changes 37.8\% less per frame than that of orthogonal frequency division multiplexing (OFDM) at 70 m/s under the same numerology, and that the resulting controller raises the average minimum SINR by 1.21 dB over equal power allocation at 40 m/s while lifting Jain's fairness index from 0.833 to 0.970, at a sum-rate cost of 24.6\% that the utility weights keep under the designer's control.}
\blue{The margin of OTFS over OFDM within the same controller widens from 0.18 dB at 10 m/s to 0.81 dB at 90 m/s, showing that the waveform contribution to the usefulness of stale SINR feedback increases with mobility.}
\end{abstract}
\vspace{-0.15cm}

\vspace{-0.25cm}
\section{Introduction}\label{Sect:Intro}

Unmanned aerial vehicle (UAV) communications are anticipated to occupy a central role in next-generation wireless networks, owing to flexible deployment, dominant line-of-sight (LoS) propagation, and broad mission scope spanning surveillance, emergency response, data collection, and temporary coverage enhancement \cite{dao2021survey}.
The price of such agility, however, is a rapidly time-varying channel: pronounced Doppler effects together with non-negligible feedback delays render timely and accurate channel state information (CSI) acquisition difficult in deployment \cite{7470933}.
Consequently, CSI-driven adaptation often operates on stale measurements and degrades in reliability whenever control or resource-allocation loops demand frequent refreshes.

Orthogonal time frequency space (OTFS) modulation has emerged as a promising waveform for high-mobility communications, as it operates in the delay-Doppler (DD) domain, where the channel representation is more structured and often more stable than in conventional orthogonal frequency division multiplexing (OFDM) systems \cite{7925924,10584089}.
Recent work has established solid foundations in waveform design, detection and transceiver development, has demonstrated the potential of multi-input multi-output OTFS in high-Doppler environments \cite{8647394,10584089}, and has begun to address UAV-oriented mobility-aware link management \cite{10663682,10063356}.
Nevertheless, a clear gap remains in delayed-feedback power control tailored to reliable and fair UAV communications under high mobility, where the focus is on weak-user protection and balanced performance rather than aggressive throughput maximization.

Delayed signal-to-interference-plus-noise ratio (SINR)-feedback control is attractive in this setting because scalar SINR is far cheaper to measure and report than full instantaneous CSI \cite{4641946}.
\blue{The classical convergence guarantees of SINR-feedback power control, from the distributed autonomous recursion of \cite{260747} to the standard interference-function framework of \cite{414651}, can in principle be applied to an OTFS UAV system. However, their conventional convergence guarantees rely on a sufficiently stationary interference mapping and a feasible target. In a high-mobility aerial link, both can change over the feedback horizon, so a stale-feedback recursion can track an outdated operating point. Optimization-based allocation such as the weighted minimum mean-square error (WMMSE) method of \cite{5756489} removes the fixed-target assumption, but requires the complete $K\times K$ link-gain matrix, for $K$ simultaneously served UAVs.}
\blue{The limitation addressed here is the validity horizon of the reported link quality rather than only its feedback overhead. The delayed SINR inherits the temporal variation of the effective link response: under OFDM it reflects a rapidly varying time-frequency response, whereas under OTFS it is formed from DD components whose physical delays and Dopplers evolve with the propagation geometry.}

\blue{Motivated by these considerations, this paper investigates delayed SINR-feedback power control for reliable and fair OTFS UAV communications under high mobility. The contributions are threefold. First, we build a multi-UAV OTFS downlink model whose DD channel, per-path Doppler and inter-UAV coupling follow from the three-dimensional air-to-ground geometry and the base-station (BS) array response, so that the OTFS and the OFDM branches share one common path set. Second, we design a controller that consumes delayed SINR alone, corrects it over the actual prediction horizon, and merges a reliability share, a fairness share and a spectral-efficiency share through the utility weights, with a feasibility and bounded-variation guarantee for the projected recursion. Third, we quantify the waveform and the controller contributions separately, and report the reliability and fairness gains together with their sum-rate cost.}

\section{System Model}
\label{sec:system}

We consider a downlink UAV communication system in which a BS serves $K$ UAV terminals over an OTFS air interface, with $\mathcal{K}=\{1,2,\ldots,K\}$ denoting the set of UAVs. The BS is equipped with a horizontal uniform linear array (ULA) of $N_{\rm t}$ elements spaced by half a wavelength, and the UAVs are served simultaneously on the same OTFS frame through per-UAV transmit beamforming, so that the interference among UAVs is the beam leakage of the array rather than an abstract coupling.

{\color{black}
\subsection{OTFS Frame and Geometry-Based Delay-Doppler Channel}
An OTFS frame carries $M$ delay bins and $N$ Doppler bins with subcarrier spacing $\Delta f$, symbol duration $T=1/\Delta f$, bandwidth $B=M\Delta f$ and frame duration $T_{\rm f}=NT$. One frame is one power-control slot, so the delay and Doppler resolutions of the receiver are $1/B$ and $1/T_{\rm f}$, respectively \cite{7925924,9392379}.
The DD channel of UAV $k$ at slot $t$ is written as
\begin{equation}
h_k(\tau,\nu;t)=\sum\limits_{p=0}^{P-1} h_{k,p}(t)\,\delta\!\left(\tau-\tau_{k,p}(t)\right)\delta\!\left(\nu-\nu_{k,p}(t)\right),
\label{eq:dd_channel}
\end{equation}
where $p=0$ indexes the LoS path and $p\ge 1$ indexes the paths reflected by static ground scatterers.
Let $\mathbf{q}_k(t)\in\mathbb{R}^3$ and $\mathbf{v}_k(t)\in\mathbb{R}^3$ denote the position and the velocity of UAV $k$, and let $\mathbf{u}_{k,p}(t)$ be the unit vector along which path $p$ arrives at the UAV. The path delay is the propagation length divided by $c$, and the Doppler shift is obtained as the projection of the velocity on the arrival direction, i.e.,
\begin{equation}
\nu_{k,p}(t)=\frac{f_c}{c}\,\mathbf{v}_k^{\sf T}(t)\,\mathbf{u}_{k,p}(t),
\label{eq:doppler_geo}
\end{equation}
with $f_c$ the carrier frequency. The Doppler of every path is therefore determined jointly by the speed and by the flight direction relative to the path, without any postulated speed-dependent factor.
The LoS power share follows a Rician factor that grows with the elevation angle $\varphi_k(t)$ of the UAV as $K_{k}(t)|_{\rm dB}=K_{\rm min}+(K_{\rm max}-K_{\rm min})\,2\varphi_k(t)/\pi$, which reproduces the elevation dependence reported for air-to-ground links \cite{8709739}, and the path amplitudes follow the distance-dependent path loss $G_0\left(d_{\rm ref}/d_{k,p}(t)\right)^{\nu_{\rm PL}}$.

\subsection{Beamforming and Effective Delay-Doppler Gains}
Let $\mathbf e_{\rm A}$ denote the unit vector along the BS ULA axis and
$\mathbf s_{k,p}^{\rm BS}(t)$ the unit departure direction of path $p$
from the BS. The corresponding normalized spatial frequency is
$\zeta_{k,p}(t)\triangleq
\mathbf e_{\rm A}^{\mathsf T}\mathbf s_{k,p}^{\rm BS}(t)$.
For half-wavelength antenna spacing, the BS array response is
$\mathbf a(\zeta)=
[1,e^{-j\pi\zeta},\ldots,e^{-j\pi(N_t-1)\zeta}]^{\mathsf T}$.
The BS steers one unit-norm beam per UAV, $\mathbf{w}_j(t)=\mathbf{a}^{*}(\zeta_{j,0}(t-d))/\sqrt{N_{\rm t}}$, matched to the LoS spatial frequency observed $d$ slots earlier, which is the same delay that affects the SINR feedback. The coefficient with which path $p$ of UAV $k$ reaches the receiver through beam $j$ is expressed as
\begin{equation}
c^{(j)}_{k,p}(t)=h_{k,p}(t)\,\mathbf{a}^{\sf T}\!\left(\zeta_{k,p}(t)\right)\mathbf{w}_j(t).
\label{eq:beam_coeff}
\end{equation}
The interference that UAV $k$ receives from beam $j$ is thus generated by the channel of UAV $k$ itself, observed through a beam pointed elsewhere. The delayed LoS direction used for beam steering is assumed available from routine navigation and beam tracking, and is not an input to the proposed power controller, whose update uses delayed SINR feedback only.

Under the standard DD-domain representation, the received OTFS grid of UAV $k$
at slot $t$ is written as
\begin{align}
&Y_k[l,q;t]
=
\sum_{j\in\mathcal{K}}
\sqrt{P_j(t)} \nonumber \\
&\sum_{l',q'}
H^{(j)}_k[l',q';t]\,
X_j\!\left([l-l']_M,[q-q']_N;t\right)
+Z_k[l,q;t],
\label{eq:otfs_io}
\end{align}
where $X_j[l,q;t]$ denotes the unit-power DD-domain symbol grid transmitted
towards UAV $j$, i.e., $\mathbb{E}[|X_j[l,q;t]|^2]=1$,
$Z_k[l,q;t]$ is the noise, and $H^{(j)}_k[l',q';t]$ follows from the paths
of \eqref{eq:dd_channel} through the beam coefficients in
\eqref{eq:beam_coeff}. A fractional delay or Doppler offset spreads a path over neighboring DD bins according to the corresponding Dirichlet kernel \cite{8424569}.
For the effective-SINR abstraction used by the power controller, the receiver combines the retained energy of the resolved DD coefficients of \eqref{eq:otfs_io} by maximum-ratio combining, so paths that share a grid cell add coherently while distinct cells add in power. Collecting into $\mathcal{B}_{l,q}$ the paths that fall in delay bin $l$ and Doppler bin $q$ gives
\begin{equation}
g^{\rm O}_{k,j}(t)=\sum\nolimits_{(l,q)}\Big|\sum\nolimits_{p\in\mathcal{B}_{l,q}}\sqrt{\epsilon_{k,p}(t)}\;c^{(j)}_{k,p}(t)\Big|^{2},
\label{eq:gain_otfs}
\end{equation}
where $\epsilon_{k,p}$ is the fraction of the path energy that the detector retains inside the delay and Doppler taps it processes. For a fractional offset $\kappa$, the retained fraction follows the Dirichlet kernel of the transform length, and the remaining share $1-\epsilon_{k,p}$ leaks onto the untreated taps and acts as self-interference $\varepsilon^{\rm O}_k(t) \approx \sum_{p=0}^{P-1}\Big[1-\epsilon_{k,p}(t)\Big]\left|c_{k,p}^{(k)}(t)\right|^2$ \cite{8424569}.

Under the same numerology, the OFDM branch transmits $N$ symbols of $M$ subcarriers. A Doppler shift $\nu$ destroys the subcarrier orthogonality, and the fraction of the path energy that remains on the intended subcarrier is $\eta_{k,p}={\rm sinc}^2(\nu_{k,p}T)$, the rest appearing as inter-carrier interference \cite{911445}. The channel frequency response of resource element $(m,n)$ is calculated as
\begin{equation}
H^{(j)}_{k}[m,n]=\sum\nolimits_{p}\sqrt{\eta_{k,p}}\,c^{(j)}_{k,p}e^{-j2\pi m\Delta f\tilde\tau_{k,p}}e^{j2\pi \nu_{k,p}nT},
\label{eq:cfr_ofdm}
\end{equation}
with $\tilde\tau_{k,p}$ the excess delay. Every resource element sees its own realization, so the scalar that the UAV can report for the frame is the capacity-effective SINR, namely the value whose rate equals the average per-resource-element rate.

\subsection{Channel-Acquisition Floor}
To compare the two interfaces under the same pilot budget $\rho_{\rm p}$ of the $MN$ grid, we use an effective coefficient-count abstraction. A time-frequency response spanning maximum excess delay $\tau^{\max}_k$ and maximum Doppler $\nu^{\max}_k$ is modeled with the per-frame count $D^{\rm F}_k=\lceil 2\tau^{\max}_kB+1\rceil\lceil 2\nu^{\max}_kT_{\rm f}+1\rceil$, renewed every frame, whereas the sparse DD representation uses the amortized load $D^{\rm O}_k=|\mathcal{B}_k|/N_{\rm re}$ over $N_{\rm re}$ frames \cite{8671740}. Under the normalized equal-pilot-energy convention of the simulations, the residual estimation-noise floor is modeled as
\begin{equation}
\sigma^2_{{\rm e},k}=\sigma^2 D_k\big/\left(\rho_{\rm p}MN\right).
\label{eq:est_floor}
\end{equation}
Here, $\mathcal{B}_k$ denotes the set of significant DD taps of UAV $k$, and $D_k=D_k^{\rm F}$ for OFDM and $D_k=D_k^{\rm O}$ for OTFS. The same convention is applied to both interfaces. Only the effective coefficient count differs, so the OFDM floor grows with Doppler through $D_k^{\rm F}$.
}

\subsection{SINR, Rate and Delayed Feedback}
\blue{Let $P_k(t)$ denote the transmit power allocated to UAV $k$, and let $g_{k,j}(t)$ and $\varepsilon_k(t)$ denote the effective beam gain and residual self-interference of the considered air interface, respectively, with $g_{k,j}(t)=g^{\rm O}_{k,j}(t)$ for OTFS. The noise power is denoted by $\sigma^2$. The received SINR of UAV $k$ is
\begin{equation}
\gamma_k(t)
=
\frac{P_k(t) g_{k,k}(t)}
{\sum\nolimits_{j\in\mathcal{K}\setminus\{k\}} P_j(t) g_{k,j}(t) + P_k(t)\varepsilon_k(t) + \sigma^2_{{\rm e},k}(t) + \sigma^2 },
\label{eq:sinr_def}
\end{equation}
and the achievable rate is $R_k(t)=\log_2(1+\gamma_k(t))$.} The transmit powers obey the total budget $P_{\max}$ and the per-UAV floor $P_{\min}$, i.e., $\sum_{k \in {\cal K}} P_k(t)\le P_{\max}$ and $P_k(t)\ge P_{\min}$, $\forall k$, and the fairness of the achieved rate distribution is assessed by Jain's index \cite{jain1984quantitative}
\begin{equation}
J(t)=
{\left(\sum\nolimits_{k \in {\cal K}} R_k(t)\right)^2}
\Big/{\Big( K\sum\nolimits_{k \in {\cal K}} R_k^2(t)\Big)}.
\label{eq:jain}
\end{equation}
Instead of assuming full and instantaneous CSI at the BS, we consider a delayed SINR-feedback model in which the BS receives from UAV $k$ the observation $\gamma_k^{\mathrm{fb}}(t)=\gamma_k(t-d)$, where $d\geq 0$ is the feedback delay in slots.

\blue{The design targets reliability and fairness while preserving a spectral-efficiency incentive. As the three quantities carry different units and ranges, each is normalized before being combined, giving the utility
\begin{equation}
U(t)
=
\omega_1\frac{\min\!\left(\gamma_{\min}(t),\gamma_{\mathrm{th}}\right)}{\gamma_{\mathrm{th}}}
+
\omega_2 J(t)
+
\omega_3\frac{\sum\nolimits_{k \in {\cal K}} R_k(t)}{K\log_2(1+\gamma_{\mathrm{ub}})},
\label{eq:utility_def}
\end{equation}
where $\gamma_{\min}(t)=\min_{k\in\mathcal{K}}\gamma_k(t)$, $\gamma_{\mathrm{th}}$ is the reliability threshold, $\gamma_{\mathrm{ub}}$ is the SINR ceiling of the link budget, and $\omega_1,\omega_2,\omega_3\ge 0$ with $\omega_1+\omega_2+\omega_3=1$.} Every term of \eqref{eq:utility_def} now lies in $[0,1]$, so the weights express a genuine preference rather than an incidental scale.
Directly maximizing \eqref{eq:utility_def} is challenging because the UAVs are interference-coupled and the BS observes only delayed SINR feedback. Therefore, this work adopts a feedback-control perspective, and the power-update problem is formulated as the design of a low-complexity mapping
\begin{equation}
\mathbf{P}(t+1)=\mathcal{F}\!\left(\mathbf{P}(t),\boldsymbol{\gamma}^{\mathrm{fb}}(t)\right),
\label{eq:control_mapping}
\end{equation}
where $\mathbf{P}(t)=[P_1(t),\ldots,P_K(t)]^{\sf T}$ and $\boldsymbol{\gamma}^{\mathrm{fb}}(t)=[\gamma_1^{\mathrm{fb}}(t),\ldots,\gamma_K^{\mathrm{fb}}(t)]^{\sf T}$.

\section{Proposed Delayed SINR-Feedback Power Control}
\label{sec:proposed}

\subsection{Prediction, Smoothing and Weighted Priority}
\blue{The feedback carries the SINR measured $d$ slots ago, so the correction has to span the same horizon. Writing the one-slot slope of the feedback as $\gamma_k^{\mathrm{fb}}(t)-\gamma_k^{\mathrm{fb}}(t-1)$, the predicted SINR is expressed as
\begin{equation}
\hat{\gamma}_k(t)
=
\Big[
\gamma_k^{\mathrm{fb}}(t)
+
\min(\lambda d,\Lambda_{\max})
\bigl(\gamma_k^{\mathrm{fb}}(t)-\gamma_k^{\mathrm{fb}}(t-1)\bigr)
\Big]_{\gamma_{\mathrm{lb}}}^{\gamma_{\mathrm{ub}}},
\label{eq:gamma_pred}
\end{equation}
where $\lambda \ge 0$ is the delay-compensation factor and $[\cdot]_{\gamma_{\mathrm{lb}}}^{\gamma_{\mathrm{ub}}}$ denotes clipping onto a prescribed SINR interval.}
\blue{The extrapolation gain grows with $d$, because a longer horizon has to be bridged, and saturates at $\Lambda_{\max}$, because a linear extrapolation is only trustworthy within the coherence time of the effective gain.}
Short-term fluctuations are then suppressed by the smoothed estimate
\begin{equation}
\bar{\gamma}_k(t)
=
\alpha \bar{\gamma}_k(t-1)
+
(1-\alpha)\hat{\gamma}_k(t),
\label{eq:gamma_smooth}
\end{equation}
with $\alpha\in[0,1)$ the smoothing coefficient.

\blue{To realize the three preferences in \eqref{eq:utility_def} through a low-complexity per-UAV update, we associate them with non-negative priority proxies. The reliability proxy is the normalized SINR shortfall, the fairness proxy is the normalized rate deficit from the swarm average, and the spectral-efficiency proxy is an SINR-to-power link-quality indicator, i.e.,
\begin{align}
s_k(t)&=\frac{\left[\gamma_{\mathrm{th}}-\bar{\gamma}_k(t)\right]^{+}}{\gamma_{\mathrm{th}}}, \nonumber
\\
f_k(t)&=\frac{\left[\bar{R}(t)-\bar{R}_k(t)\right]^{+}}{\bar{R}(t)},
\nonumber\\
e_k(t)&=\frac{\gamma_k^{\mathrm{fb}}(t)}{P_k(t-d)},
\label{eq:priority_components}
\end{align}
where $\bar{R}_k(t)=\log_2(1+\bar{\gamma}_k(t))$ and $\bar{R}(t)$ is its average over the swarm. Since the BS knows the power $P_k(t-d)$ that it transmitted $d$ slots earlier, $e_k(t)$ is formed from the delayed SINR and the BS power history and therefore requires no additional channel-quality feedback.}

\blue{Every component is turned into a share that sums to $K$, so that a value of one means no preference and an inactive component returns the neutral vector, and the shares are mixed by the utility weights, i.e.,
\begin{equation}
\tilde{q}_k(t)=\sum\nolimits_{i=1}^{3}\omega_i\left(\frac{K x^{(i)}_k(t)}{\sum_{n\in\mathcal{K}} x^{(i)}_n(t)}-1\right),
\label{eq:centered_score}
\end{equation}
with $x^{(1)}_k=s_k$, $x^{(2)}_k=f_k$, $x^{(3)}_k=e_k$, and with the $i$-th bracket set to zero when $\sum_n x^{(i)}_n=0$.}
By construction $\sum_k\tilde{q}_k(t)=0$, so the update redistributes power instead of scaling the whole vector, and the weights of \eqref{eq:utility_def} enter the closed loop directly.

\subsection{Projected Multiplicative Update}
The power update is performed in the log-power domain as
\begin{equation}
u_k(t)
=
\Big[\mu_t \tilde{q}_k(t)
+
\beta\bigl(\tilde{q}_k(t)-\tilde{q}_k(t-1)\bigr)\Big]_{-\Delta_{\max}}^{+\Delta_{\max}},
\label{eq:update_signal}
\end{equation}
where $\mu_t=\mu/(1+\eta d)$ is the delay-aware step size, $\mu>0$ is the nominal step size, $\beta\ge 0$ is a momentum coefficient, $\eta\ge 0$ controls the delay sensitivity, and $\Delta_{\max}>0$ is the maximum log-power increment.
The tentative allocation $\tilde{P}_k(t+1)=P_k(t)\exp(u_k(t))$ may violate the power budget, so the vector $\tilde{\mathbf{P}}(t+1)$ is projected onto the feasible set
\begin{equation}
\mathcal{P}
=
\left\{
\mathbf{P}\in\mathbb{R}_+^K:
\sum\limits_{k \in {\cal K}} P_k\le P_{\max},
\,
P_k\ge P_{\min},\, \forall k
\right\},
\label{eq:feasible_set}
\end{equation}
which yields the final update $\mathbf{P}(t+1)=\Pi_{\mathcal{P}}(\tilde{\mathbf{P}}(t+1))$ with the Euclidean projection $\Pi_{\mathcal{P}}(\cdot)$ \cite{Condat2016}, as summarized in Algorithm~\ref{alg:proposed_controller}.

{\color{black}
\begin{proposition}
\label{prop:stability}
Let $\mathbf{P}(1)\in\mathcal{P}$ and let $\mathbf{P}(t+1)=\Pi_{\mathcal{P}}(\mathbf{P}(t)\odot e^{\mathbf{u}(t)})$ with $\|\mathbf{u}(t)\|_\infty\le\Delta_{\max}$. Then, for every $t\ge 1$,
\emph{(i)} $\mathbf{P}(t)\in\mathcal{P}$;
\emph{(ii)} the normalized power-update residual obeys $\rho_P(t)\triangleq\|\mathbf{P}(t+1)-\mathbf{P}(t)\|_1/P_{\max}\le\sqrt{K}\,(e^{\Delta_{\max}}-1)$, uniformly in the feedback delay and in the channel realization;
\emph{(iii)} under a frozen feedback state, any feasible allocation satisfying $\tilde{\mathbf{q}}(t)=\tilde{\mathbf{q}}(t-1)=\mathbf{0}$ is a fixed point of the one-step recursion.
\end{proposition}
\begin{IEEEproof}
Claim \emph{(i)} holds because $\Pi_{\mathcal{P}}$ maps into $\mathcal{P}$. For \emph{(ii)}, clipping in \eqref{eq:update_signal} gives $|\tilde{P}_k(t+1)-P_k(t)|=P_k(t)|e^{u_k(t)}-1|\le P_k(t)(e^{\Delta_{\max}}-1)$. As $\mathcal{P}$ is closed and convex, $\Pi_{\mathcal{P}}$ is nonexpansive and $\Pi_{\mathcal{P}}(\mathbf{P}(t))=\mathbf{P}(t)$ by (i), hence $\|\mathbf{P}(t+1)-\mathbf{P}(t)\|_2\le\|\tilde{\mathbf{P}}(t+1)-\mathbf{P}(t)\|_2\le(e^{\Delta_{\max}}-1)\|\mathbf{P}(t)\|_2$. Combining $\|\mathbf{x}\|_1\le\sqrt{K}\|\mathbf{x}\|_2$ with $\|\mathbf{P}(t)\|_2\le\|\mathbf{P}(t)\|_1\le P_{\max}$ gives the bound. For \emph{(iii)}, $\tilde{\mathbf{q}}(t)=\tilde{\mathbf{q}}(t-1)=\mathbf{0}$ makes $\mathbf{u}(t)=\mathbf{0}$, so $\tilde{\mathbf{P}}(t+1)=\mathbf{P}(t)\in\mathcal{P}$ and the projection is inactive. The proof is concluded.
\end{IEEEproof}
Proposition~\ref{prop:stability} therefore guarantees feasibility and a uniform bound on the slot-to-slot power variation, while a zero-priority allocation is stationary under a frozen feedback state. Asymptotic convergence is not claimed for the time-varying UAV channel, and the corresponding bounded-tracking behavior is evaluated in Section~\ref{sec:results}.

\blue{Equal power is used until the first delayed SINR report is available. At the first controller update, the previous feedback sample is initialized by the current report, i.e., $\gamma_k^{\mathrm{fb}}(t-1)=\gamma_k^{\mathrm{fb}}(t)$, which gives a zero initial prediction slope, and $P_k(t-d)$ in \eqref{eq:priority_components} is well defined because that update occurs only after the corresponding transmitted power has entered the BS power history.}

\begin{algorithm}[t]
\small
\caption{\small Proposed Delayed SINR-Feedback Power Control}
\label{alg:proposed_controller}
\KwIn{$P_{\max}$, $P_{\min}$, $\gamma_{\mathrm{th}}$, $d$, $\Delta_{\max}$, $\mu$, $\beta$, $\eta$, $\lambda$, $\Lambda_{\max}$, $\alpha$, $\gamma_{\mathrm{lb}}$, $\gamma_{\mathrm{ub}}$, $\{\omega_i\}$.}
Initialize $P_k(1)=P_{\max}/K$, $\bar{\gamma}_k(0)=\gamma_{\mathrm{th}}$, $\tilde{q}_k(0)=0$, $\forall k$\;
\For{$t=1,2,\ldots,T-1$}{
    Receive delayed SINR feedbacks $\gamma_k^{\mathrm{fb}}(t)$, $\forall k$\;
    \color{black}
    Predict $\hat{\gamma}_k(t)$ by \eqref{eq:gamma_pred} and smooth to $\bar{\gamma}_k(t)$ by \eqref{eq:gamma_smooth}, $\forall k$\;
    Form the reliability, fairness and spectral-efficiency components by \eqref{eq:priority_components}, $\forall k$\;
    Mix them into the centered priority $\tilde{q}_k(t)$ by \eqref{eq:centered_score}, $\forall k$\;
    Compute $u_k(t)$ by \eqref{eq:update_signal} and $\tilde{P}_k(t+1)=P_k(t)e^{u_k(t)}$, $\forall k$\;
    Project onto \eqref{eq:feasible_set}: $\mathbf{P}(t+1)=\Pi_{\mathcal{P}}(\tilde{\mathbf{P}}(t+1))$\;
}
\end{algorithm}

\color{black}

\subsection{Computational Complexity}\label{Sect:Complexity}
\blue{Prediction, smoothing, component construction and the log-domain update require $\mathcal{O}(K)$ operations each, and the projection is dominated by one sort, so the per-slot cost is $\mathcal{O}(K\log K)$ with $\mathcal{O}(K)$ memory. Table~\ref{tab:complexity} places it next to the benchmarks with the information each consumes and the per-update running time measured on an Intel Core i9 platform in MATLAB R2024a, averaged over all slots of the campaign.}

\section{Numerical Results and Discussion}
\label{sec:results}

\blue{A BS with $N_{\rm t}=32$ array elements serves $K=6$ UAVs that fly at 120 m over a $180^\circ$ azimuth sector at horizontal ranges of 100--600 m. Both interfaces use $f_c=5.8$ GHz, $\Delta f=15$ kHz, $M=128$ and $N=64$, giving $B=1.92$ MHz, $T_{\rm f}=4.27$ ms, a delay resolution of 521 ns and a Doppler resolution of 234 Hz, and both spend $\rho_{\rm p}=5\%$ of the grid on pilots with $N_{\rm re}=10$. The channel uses $\nu_{\rm PL}=2.2$, a Rician factor from 3 dB at grazing elevation to 12 dB at zenith \cite{8709739}, and five static ground scatterers per UAV placed 50--400 m from its ground projection. Each UAV keeps a nominal speed of 40 m/s unless stated otherwise, with a per-slot heading increment of standard deviation $2^\circ$, a first-order speed variation of standard deviation 3\%, and reflection at the annulus boundaries. Powers are normalized by $P_{\max}=1$ with $P_{\min}=0.02$, the transmit SNR is $P_{\max}/\sigma^2=10$ dB, and $d=2$ slots. 
The controller uses $\gamma_{\mathrm{th}}=6$ dB, $\lambda=0.35$ with
$\Lambda_{\max}=1.2$, $\alpha=0.6$, $\mu=0.05$, $\beta=0.1$,
$\eta=0.25$, and $\Delta_{\max}=0.18$, selected from a stable
operating region. The utility weights are set to
$(\omega_1,\omega_2,\omega_3)=(0.5,0.3,0.2)$, corresponding to the
knee of the reliability--throughput tradeoff.
Every point averages 60 topologies of 400 frames with the first 40\% discarded as transient, all schemes sharing the same draws. The curves plot trial means, and the 95\% confidence half-width stays below 1.41 dB on any SINR curve, 0.069 on any fairness curve and 0.93 bit/s/Hz on any rate curve, and equals 0.40 dB and 0.011 for the proposed controller at the operating point.}

\blue{Four benchmarks are considered besides the proposed controller running on OFDM. Equal power allocation splits $P_{\max}$ evenly. Instantaneous gain-based (GB) allocation sets $P_k\propto g_{k,k}(t)$ and therefore needs instantaneous CSI, so it acts as a genie reference. Delayed SINR balancing (SB) is the damped Foschini--Miljanic recursion \cite{260747,414651}, which drives every UAV towards the geometric mean of the delayed feedback. Delayed-CSI WMMSE maximizes the sum rate on the delayed $K\times K$ gain matrix with $I_{\rm w}=8$ inner iterations \cite{5756489}.}


\begin{table}[t]
\color{black}
\centering
\caption{\color{black}Per-slot information and computational cost at
$K=6$ and $I_{\rm w}=8$.}
\label{tab:complexity}
\renewcommand{\arraystretch}{1.12}
\setlength{\tabcolsep}{2.5pt}
\footnotesize
\begin{tabular*}{\columnwidth}{@{\extracolsep{\fill}}lccc@{}}
\toprule
\textbf{Scheme}
& \textbf{Input}
& \textbf{Complexity}
& \textbf{Time ($\mu$s)} \\
\midrule
Proposed
& $K$ delayed SINRs
& $\mathcal{O}(K\log K)$
& 11.1 \\

Equal power
& None
& $\mathcal{O}(1)$
& 2.5 \\

Instant. GB
& $K$ inst. gains
& $\mathcal{O}(K\log K)$
& 3.8 \\

Delayed SB \cite{260747,414651}
& $K$ delayed SINRs
& $\mathcal{O}(K\log K)$
& 6.6 \\

WMMSE \cite{5756489}
& $K^2$ delayed gains
& $\mathcal{O}(I_{\rm w}K^2)$\textsuperscript{$\dagger$}
& 54.5 \\
\bottomrule
\end{tabular*}
\footnotesize{\textsuperscript{$\dagger$}$I_{w}$ denotes the number of WMMSE iterations.}
\end{table}

\begin{figure}[t]
	\centering
    \captionsetup{justification=raggedright}
	\includegraphics[width=\linewidth]{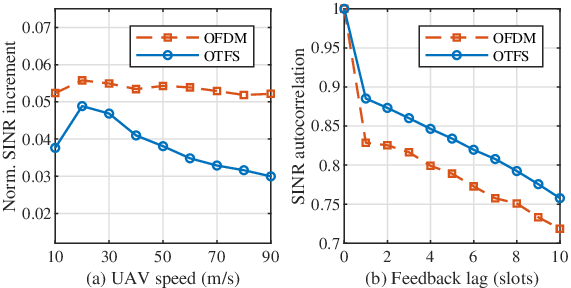}
	\caption{\blue{Air-interface behavior under equal power allocation, without any controller: (a) normalized per-frame SINR increment versus UAV speed, and (b) autocorrelation of the effective SINR versus feedback lag at 70 m/s.}}
    \label{fig:waveform}
\vspace{-0.2cm}
\end{figure}

\blue{Fig.~\ref{fig:waveform} isolates the waveform contribution by evaluating both interfaces under equal power allocation, without any controller. In Fig.~\ref{fig:waveform}(a), the normalized per-frame SINR increment of OFDM stays close to 0.053 over the whole speed range, whereas the OTFS increment falls from 0.049 at 20 m/s to 0.030 at 90 m/s. The gap stays within the confidence half-width of the trial mean up to $40$ m/s and becomes more pronounced beyond that point, with reductions of $37.8\%$ at $70$ m/s and $42.5\%$ at $90$ m/s. 
The separation becomes pronounced once the path-Doppler spread is comparable to or exceeds the DD resolution $1/T_{\rm f}=234$ Hz, allowing more of the multipath structure to be resolved on the DD grid.
In the time domain, Fig.~\ref{fig:waveform}(b) gives an OTFS autocorrelation of 0.873 at a lag of two slots against 0.825 for OFDM at 70 m/s, so the stale report describes the present link more closely under OTFS. The residual self-interference is also smaller, 0.019\% of the desired gain against 0.402\% for OFDM, but at $\Delta f=15$ kHz both are too small to drive the difference.}

\begin{figure}[t]
	\centering
    \captionsetup{justification=raggedright}
	\includegraphics[width=\linewidth]{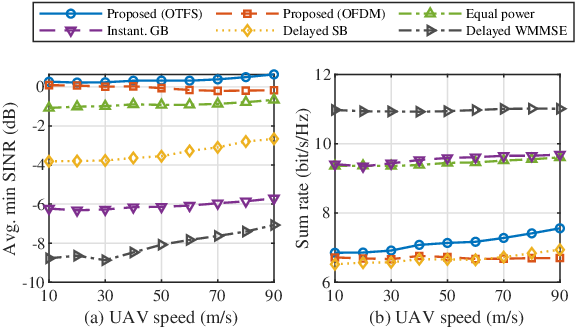}
	\caption{\blue{Impact of UAV speed on (a) the average minimum SINR and (b) the sum rate.}}
    \label{fig:speed}
\vspace{-0.2cm}
\end{figure}

Fig.~\ref{fig:speed} shows the impact of the UAV speed on the average minimum SINR and on the sum rate. As shown in Fig.~\ref{fig:speed}(a), the proposed OTFS controller attains the highest minimum SINR over the whole range. \blue{At 40 m/s, it reaches 0.31 dB, which is 1.21 dB above equal power allocation, 3.95 dB above delayed SINR balancing, 6.48 dB above instantaneous GB allocation and 8.79 dB above delayed-CSI WMMSE, the last two ending lowest because both concentrate power on the UAVs that already enjoy the strongest links, which the shortfall share of \eqref{eq:priority_components} prevents. The advantage of OTFS over OFDM inside the same controller grows monotonically with mobility, from 0.18 dB at 10 m/s to 0.81 dB at 90 m/s, following the separation seen in Fig.~\ref{fig:waveform}(a). 
All curves drift upwards modestly, because a faster UAV traverses a larger portion of the cell within the observation window, and the drift leaves the ordering unchanged. Fig.~\ref{fig:speed}(b) reports the price of that protection: 7.08 bit/s/Hz at 40 m/s against 9.38 bit/s/Hz for equal power allocation, a loss of 24.6\%, while delayed-CSI WMMSE attains the highest sum rate, 10.92 bit/s/Hz, as its objective dictates.}

\begin{figure}[t]
	\centering
    \captionsetup{justification=raggedright}
	\includegraphics[width=0.958\linewidth]{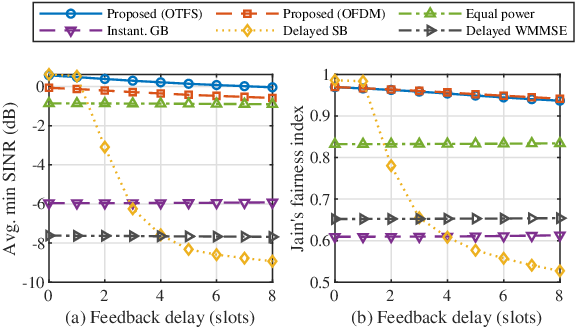}
	\caption{\blue{Impact of the feedback delay at 70 m/s on (a) the average minimum SINR and (b) Jain's fairness index.}}
    \label{fig:delay}
\vspace{-0.2cm}
\end{figure}

\blue{Fig.~\ref{fig:delay} evaluates the impact of feedback delay at a UAV speed of 70 m/s, which represents a high-mobility operating condition. In Fig.~\ref{fig:delay}(a) the proposed OTFS controller degrades gracefully, from 0.57 dB at $d=0$ to $-0.04$ dB at $d=8$, and stays above equal power allocation throughout. Delayed SINR balancing behaves in the opposite way: it is the closest competitor at $d\le 1$, then collapses to $-3.10$ dB at $d=2$ and to $-8.93$ dB at $d=8$, because its recursion maps the delayed SINR ratio directly into the power update without the delay-dependent step reduction and clipping used in \eqref{eq:update_signal}. Fig.~\ref{fig:delay}(b) confirms the same mechanism on fairness, the proposed index falling only from 0.969 to 0.937 while the balancing benchmark falls from 0.986 to 0.527. Equal power allocation is delay-independent by construction and instantaneous GB allocation reads the current gains. 
Delayed-CSI WMMSE is itself a stale-information benchmark, yet it moves only from $-7.62$ dB to $-7.69$ dB: it is already the lowest curve at $d = 0$, and the slow evolution of the effective OTFS link in Fig.~\ref{fig:waveform}(b) keeps its delayed input representative, so its weakness is its sum-rate objective rather than the staleness of its input.}

\begin{figure}[t]
	\centering
    \captionsetup{justification=raggedright}
	\includegraphics[width=\linewidth]{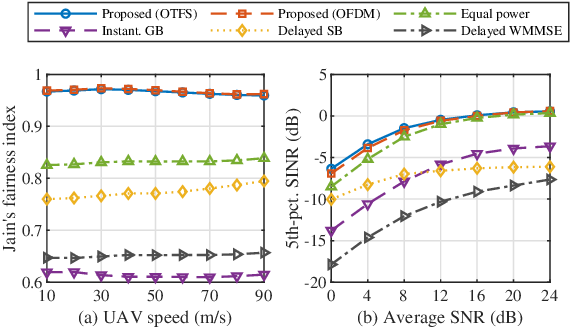}
	\caption{\blue{(a) Jain's fairness index versus UAV speed and (b) 5th-percentile SINR versus average SNR.}}
    \label{fig:fair}
\vspace{-0.2cm}
\end{figure}

Fig.~\ref{fig:fair} shows Jain's fairness index versus the UAV speed and the 5th-percentile SINR versus the average SNR. \blue{In Fig.~\ref{fig:fair}(a) the proposed controller holds the index between 0.960 and 0.972, against 0.833 for equal power allocation, 0.771 for delayed SINR balancing and 0.610 for instantaneous GB allocation. The high level follows from the reliability- and fairness-oriented shares in \eqref{eq:centered_score}, which shift power toward disadvantaged UAVs; under frozen feedback, Proposition~\ref{prop:stability}(iii) identifies stationarity when the aggregate centered priority vanishes. Jain's index is also computed on rates, whose logarithm compresses the residual SINR spread. Fairness is close to flat in speed, so its ordering is set by the allocation rule rather than by mobility. In Fig.~\ref{fig:fair}(b) the tail improves with SNR and saturates beyond about 16 dB, where the residual interference rather than the noise limits the weakest UAV, and the proposed controller keeps the highest 5th-percentile SINR throughout, by 0.52 dB over equal power allocation and 5.35 dB over instantaneous GB allocation at 12 dB.}

\begin{figure}[t]
	\centering
    \captionsetup{justification=raggedright}
	\includegraphics[width=\linewidth]{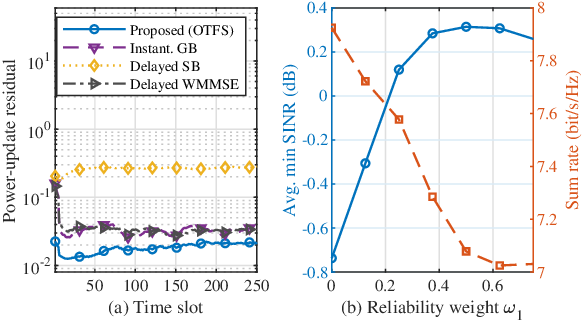}
	\caption{\blue{(a) Power-update residual versus time slot, smoothed by a nine-slot moving average, and (b) effect of the reliability weight $\omega_1$, with $\omega_2=0.6(1-\omega_1)$ and $\omega_3=0.4(1-\omega_1)$, on the average minimum SINR and on the sum rate.}}
    \label{fig:conv_weight}
\vspace{-0.2cm}
\end{figure}

\blue{Fig.~\ref{fig:conv_weight} reports the closed-loop behavior and the effect of the utility weights. In Fig.~\ref{fig:conv_weight}(a), after the initial transient, the nine-slot-smoothed residual of the proposed loop remains in the band $0.013$--$0.024$, about twenty times below the bound $\sqrt{K}(e^{\Delta_{\max}}-1)=0.483$ of Proposition~\ref{prop:stability}(ii), whereas delayed SINR balancing keeps a residual near $0.26$ because it repeatedly re-scales the power vector using stale measurements. Fig.~\ref{fig:conv_weight}(b) shows the minimum-SINR and sum-rate exchange. Raising $\omega_1$ from 0 to 0.5 lifts the average minimum SINR from $-0.74$ dB to 0.31 dB while the sum rate falls from 7.93 to 7.08 bit/s/Hz. Over the same sweep Jain's index, not plotted for readability, increases from 0.756 to 0.970. The metrics change only slightly beyond $\omega_1=0.5$, confirming that the selected operating point lies near the knee of the reliability--throughput tradeoff.}

\section{Conclusion}
\label{sec:conclusion}
This work developed delayed SINR-feedback power control for reliable and fair high-mobility OTFS UAV links on a geometry-based DD channel with delayed beam directions.
\blue{The predictor, the weighted priority rule and the projected multiplicative update adapt power from delayed SINR alone, and Proposition~\ref{prop:stability} guarantees feasibility and uniformly bounded per-slot power variation.}
\blue{On a common physical path set, OTFS shows 37.8\% lower per-frame SINR variation than OFDM at 70~m/s, and the controller gains 1.21~dB of average minimum SINR and lifts Jain's index from 0.833 to 0.970 over equal power at 40~m/s, at a 24.6\% sum-rate cost set by the utility weights. Its OTFS-over-OFDM margin grows from 0.18~dB at 10~m/s to 0.81~dB at 90~m/s, and the delay sweep shows far stronger robustness than classical delayed SINR balancing.}
Future work will consider MIMO-OTFS, joint beam and power control, and quantized SINR feedback.

\bibliographystyle{IEEEtran}
\bibliography{reference}
\end{document}